\documentclass[12pt]{article}
\usepackage{amsmath}
\usepackage{amssymb}
\usepackage[letterpaper, portrait, margin=1in]{geometry}
\usepackage{graphicx}
\usepackage{setspace}
\usepackage{natbib,hyperref}
\usepackage{url}

\newtheorem{theorem}{Theorem}
\newtheorem{case}{Case}

\newtheorem{definition}{Definition}
\newtheorem{example}{Example}

\newtheorem{lemma}{Lemma}

\newtheorem{proposition}{Proposition}
\newtheorem{remark}[theorem]{Remark}

\newenvironment{proof}[1][Proof]{\noindent\textbf{#1.} }{\ \rule{0.5em}{0.5em}}
\date{}
\DeclareMathOperator\inte{int} 
\begin{document}

\title{Risk and Intertemporal Preferences over Time Lotteries}
\author{Minghao Pan \thanks{UCLA. email: \href{mailto:minghaopan@ucla.edu}{minghaopan@ucla.edu}}}
\maketitle
\begin{abstract}
This paper studies relations among axioms on individuals' intertemporal
choices under risk. The focus is on Risk Averse over Time Lotteries (RATL),
meaning that a fixed prize is preferred to a lottery with the same monetary
prize but a random delivery time. Though surveys and lab experiments documented RATL
choices, Expected Discounted Utility cannot accommodate any RATL. This paper's contribution is two-fold. First, under a very weak form of Independence,
we generalize the incompatibility of RATL with two axioms about
intertemporal choices: Stochastic Impatience (SI) and No Future Bias. Next, we
prove a representation theorem that gives a class of models satisfying RATL
and SI everywhere. This illustrates that there is no fundamental conflict
between RATL and SI, and leaves open possibility that RATL behavior is
caused by Future Bias.\\
\vspace{0in}
\end{abstract}

\let\thefootnote\relax\footnotetext{%
I would like to express my deepest appreciation to Jay Lu for his guidance and encouragement throughout the project. I would like to thank
Jiayin Zhai, as well as seminar participants at Caltech, for their helpful comments. I am very grateful to Victoria Wang whose feedback significantly improved the readability of the paper. I appreciate Huihuang Zhu for providing technical support. The author declares that he has no relevant or material financial interests that relate to the research described in this paper.}

\section{Introduction}

The study of intertemporal choices under risk has a long history.
Traditionally, the Expected Discounted Utility (EDU) model is widely used to
predict individuals' intertemporal choices under risk thanks to its
simplicity and tractable function form. EDU inherits the shortcomings of
both Expected Utility and Discounted Utility, so absurdities caused by EDU
have motivated the search for alternative models. Popular extensions of
Discounted Utility include the quasi-hyperbolic discounting (\cite{PP68}) and the generalized hyperbolic discounting (\cite{LP92}), not to mention the numerous generalizations to Expected Utility.

An absurdity about EDU concerns time lotteries - those with fixed monetary
prizes and random delivery time. Time lotteries arise naturally in
practice, for example, when a manufacturer sets a fixed price for its
product but is unsure about when the product will be sold, or when a
passenger waits for a taxi that has fixed fees. Borrowing from the study of
risk aversion over random monetary payoffs, we define a choice to be Risk Seeking
over Time Lotteries (RSTL) if a time lottery is preferred to its certain
counterpart -- the lottery with the same monetary prize and a fixed delivery
time as the expected delivery time of the time lottery. Otherwise, the
choice is said to be Risk Averse over Time Lotteries (RATL).

It has been observed that EDU as well as commonly used Expected Utility
models with alternative discount functions predict that decision makers are
always RSTL (\cite{CV03}, \cite{OO07}).
Whenever a preference is represented by $V(p)=\mathbb{E}_{p}[D(t)v(x)]$,
where $D$ can be regarded as the discount function on time and $v$ is the
evaluation of monetary prizes, the convexity of $D$ implies RSTL. On the other
hand, surveys and lab experiments have found evidence in support of RATL in some
situations (\cite{CV03}; \cite{OO07}; \cite{DDGO20}).

What is more puzzling is an incompatibility result by \cite{DDGO20}. They put forward the assumption of Stochastic Impatience (SI); this property requires that, when pairing two monetary prizes with two delivery times, individuals like to match the higher monetary prize with the earlier delivery time. Moreover, they relaxed the assumption of Stationarity in time to No
Future Bias, which means the marginal rate of substitution of time for money is decreasing in time. Assuming Independence, No Future Bias, and other
standard assumptions, DeJarnette et al. found that SI implies RSTL.

In this article, we demonstrate the incompatibility among a weaker notion of
RATL, SI, and No Future Bias when Independence is further relaxed.
Specifically, this illustrates the incompatibility of RATL, SI, and No Future
Bias beyond Expected Utility and Generalized Local Bilinear Utility
considered in \cite{DDGO20}. The weak independence we endow is
strictly weaker than both the Axiom of Degenerate Independence (\cite{GKP92}) and the Certainty Independence Axiom (\cite{GS89}).

Given that our incompatibility result does not assume Independence, we are
motivated to reconcile SI with RATL by dropping No Future Bias. We prove a
representation theorem that gives conditions for an Expected Utility model to
satisfy SI and RATL everywhere. The representation theorem results in a
functional form of 
\begin{equation*}
V(p)=\mathbb{E}_{p}[\phi (D(t)v(x))].
\end{equation*}%
This functional form includes Generalized EDU considered in \cite{DDGO20} as a special case, in which $D(t)=\beta ^{t}$. Moreover, compared
to the standard model $V(p)=\mathbb{E}_{p}[D(t)v(x)]$, the additional
curvature provided by $\phi $ governs risk attitude, which makes RATL easier
to be attained. In particular, this representation theorem shows there is no
fundamental contradiction between SI and RATL, even within Expected Utility.
We conclude that RATL may be accommodated by allowing some Future Bias.

\section{Related Literature}

\cite{CV03} were the first to notice that EDU implies RSTL as
long as the discount rate is positive. This result was generalized by \cite{OO07}, who demonstrated RSTL for Discounted Utility
models with convex discount functions.

SI was first studied in \cite{DDGO20}, and in the same paper they proved the incompatibility of RATL with SI and No Future Bias in the
Expected Utility models and Generalized Local Bilinear Utility (GLBU) models. They discussed two solutions to accommodate RATL.
The first one is to drop  ``Intertemporal''
Independence and keep SI; the Dynamic Ordinal Certainty Equivalent model (\cite{Se78}, \cite{SS78}) may reconcile SI with RATL because
of the  ``order of aggregation.'' The second
solution is to drop SI. In this case, the Generalized Expected Discounted
Utility model

\begin{equation*}
V(p)=\mathbb{E}_{p}[\phi \left( \beta ^{t}u(x)\right) ]
\end{equation*}%
is able to satisfy RATL when the curvature of $\phi $ is large enough.

The surveys and lab experiments conducted by \cite{CV03}, \cite{OO07} and \cite{DDGO20} to test RATL usually let participants choose from a time lottery and its certain counterpart. \cite{CV03} conducted a survey among business owners and managers and found that a majority chose the risky option. The survey conducted by \cite{OO07} among 146 business managers revealed a preference for the certain option, especially when there is a small probability of very early arrival. \cite{DDGO20} run a lab experiment among 196 subjects, which concluded that the majority of subjects made the RATL choices in
the majority of questions. As for SI, \cite{LMQZ20} found evidence in favor of SI, especially among those whose choices are in
accordance with Impatience and Lottery Equivalence.

No Future Bias is the key axiom that separates our results from previous
studies of RATL. This notion, first proposed by \cite{Pr04} under the name  ``decreasing impatience,'' has been controversial. One piece of evidence supporting No Future Bias is
preference reversals (\cite{GFM94}, \cite{KH95}), which means subjects initially prefer a small but early reward
over a large but late reward, and after equally delaying the two prizes, the
preference relation is reversed.

On the other hand, some recent lab experiments and surveys find evidence
against No Future Bias, especially when the time delay between the first and the second test is small. \cite{SO09} and \cite{ABRW10} documented the
converse of preference reversals: a smaller and sooner prize is more likely to
be preferred to a larger and later prize after equally delaying the two prizes. \cite{Ta11} took a non-parametric approach and found that most subjects are Future Biased.

When assuming separability between time and money, decision makers' discount rate in a given time is independent of the amount of prize, and this is characterized by a discount
function $D(t)$. In this case, No Future Bias is equivalent to a declining
discount rate, also called ``hyperbolic
discounting.'' Previous findings of time inconsistencies
seemed so strong that hyperbolic discounting soon became popular. Two
commonly used hyperbolic discount functions are the standard one, $D(t)=\frac{1}{1+kt},k>0$, and the
quasi-hyperbolic discounting, $D(t)=\beta \delta ^{t},\delta <1$, if $t>0$. There have
been numerous tests of hyperbolic discounting, yielding inconsistent and
inconclusive results. Additional discussions can be found in \cite{Ru03}, \cite{BBS10}, and \cite{Ha15}.

Many utility representations have been proposed other than the standard one, $%
V(p)=\mathbb{E}_{p}[D(t)v(x)]$. \cite{FR82} used a
``utility independence'' condition which yields $%
V(p)=\mathbb{E}_{p}[D(t)v(x)+\omega (t)]$. \cite{MPST21} proposed a utility representation on time lotteries, assuming a
strong form of stationarity. However, the standard model always satisfies SI and the model of \cite{FR82} in general does not have a clean characterization of SI. The representation of \cite{DDGO20}, $V(p)=%
\mathbb{E}_{p}[\phi (\beta ^{t}u(x))]$, has the advantage that SI is captured by the curvature of $\phi$. In fact, their model is a special case of ours; our model allows any reasonable discount function $D(t)$ in place of $\beta^t$, thus providing more flexibility. Many more
papers, including \cite{Ep83}, \cite{Bl20}, and \cite{DGO20}, study preferences over lotteries over consumption streams.

\section{Setup}

Let $X\subset \mathbb{R}^{+}$ be a closed interval of monetary prizes and $%
T\subset \mathbb{R}^{+}\cup \left\{ 0\right\} $ be a closed interval of
delivery dates, both possibly unbounded. A pair $\left( x,t\right) $ means $%
x $ amount of money is received at time $t$. Denote by $\Delta \left(
X\times T\right) $ the set of simple lotteries on $X\times T$ so these
lotteries are random in both quantity and time. We will use $p$ to denote a
lottery in $\Delta \left( X\times T\right) $. This paper studies preference
relations $\gtrsim $ over $\Delta \left( X\times T\right) $.

We can fix the quantity dimension and only consider variation over the
delivery time. For any $x\in X$,\ we say $p_{x}\in \Delta \left( X\times
T\right) $ is a \textbf{time lottery with prize} $x$ if $p_{x}$ is supported
on $\left\{ x\right\} \times T$. Let 
\begin{equation*}
\bar{t}_{p_{x}}:=\sum t\cdot p_{x}\left( x,t\right)
\end{equation*}%
be the expected arrival time of a time lottery.

In analogy to the standard risk aversion over random monetary prizes, we
introduce the notion of risk aversion in the time dimension.

\begin{definition}
\ A preference relation $\gtrsim $ over $\Delta \left( X\times T\right) $ is 
\textbf{Risk Averse over Time Lotteries} (RATL) if for any time lottery $%
p_{x}$, $\delta _{\left( x,\bar{t}_{p_{x}}\right) }\gtrsim p_{x}$.
Similarly, a relation $\gtrsim $ is \textbf{Risk Seeking over Time Lotteries
(}RSTL) if for any time lottery $p_{x}$, $\delta _{\left( x,\bar{t}%
_{p_{x}}\right) }\lesssim p_{x}$.
\end{definition}

By restricting attention to binary lotteries with two events of probability $%
\frac{1}{2}$, we have a weaker property that plays an important role in the
following discussions.

\begin{definition}
We say $\gtrsim $ satisfies \textbf{weak RATL} if for any $x\in X$ and $%
t_{1},t_{2}\in T$, 
\begin{equation*}
\delta _{(x,\frac{t_{1}+t_{2}}{2})}\gtrsim \frac{1}{2}\delta _{(x,t_{1})}+%
\frac{1}{2}\delta _{(x,t_{2})}.
\end{equation*}%
Similarly, $\gtrsim $ is\textbf{\ weak RSTL} if the above holds with $%
\lesssim $ in place of $\gtrsim $.
\end{definition}

\section{Incompatibility Theorem}

The standard Expected Discounted Utility (EDU) model evaluates a lottery $%
p\in \Delta \left( X\times T\right) $ by 
\begin{equation*}
V(p)=\mathbb{E}_{p}[\beta ^{t}v(x)],
\end{equation*}%
where $\beta \in (0,1)$ and $v$ is non-negative and strictly increasing.
Despite the evidence of RATL in some situations, EDU predicts individuals
are always RSTL. A slightly more general model has the functional form 
\begin{equation*}
V(p)=\mathbb{E}_{p}[D(t)v(x)],
\end{equation*}%
where $D$ is often called the discount function. Using Jensen's inequality,
we easily deduce that this model obeys RSTL if $D$ is convex, which is the
case of the generalized quasi-hyperbolic discount function $D(t)=(1+\alpha
t)^{-\beta /\alpha }$, $\alpha ,\beta >0$.

This discrepancy between the observed RATL behavior and no accommodation of RATL
in commonly used theoretical models is further generalized by \cite{DDGO20}. They introduced the notion of Stochastic Impatience.

\begin{definition}
A preference\ relation $\gtrsim $ satisfies \textbf{Stochastic Impatience }%
(SI) if for any $t_{1}<t_{2}\in T$ and $x_{1}>x_{2}\in X$, we have%
\begin{equation*}
\frac{1}{2}\delta _{(x_{1},t_{1})}+\frac{1}{2}\delta _{(x_{2},t_{2})}\gtrsim 
\frac{1}{2}\delta _{(x_{1},t_{2})}+\frac{1}{2}\delta _{(x_{2},t_{1})}.
\end{equation*}%
In other words, individuals like to match the earlier arrival time with the
higher prize.
\end{definition}

Clearly, if $\gtrsim $ can be represented by 
\begin{equation*}
V(p)=\mathbb{E}_{p}[D(t)v(x)],
\end{equation*}%
where $D\in (0,1]$ is monotone decreasing and $v$ is non-negative and
monotone increasing, then $\gtrsim $ satisfies SI.

Dynamic inconsistency observed in experiments initiated search for alternative
models, and many forms of hyperbolic discount functions emerged. \cite{Pr04} suggested it is one common property behind these models that drives dynamic inconsistency, which he termed as ``decreasing
impatience.'' We instead use the name coined by \cite{DDGO20}.

\begin{definition}
A preference\ relation $\gtrsim $ satisfies \textbf{No Future Bias} if for
all $x,y\in X$, $s,t\in T$ with $t<s$ and $\tau >0$, 
\begin{equation*}
\delta _{(x,t)}\sim \delta _{(y,s)}\Longrightarrow \delta _{(x,t+\tau
)}\lesssim \delta _{(y,s+\tau )}.
\end{equation*}%
For some $x,y,s,t$ as above, we say the choice has \textbf{Future Bias} if the
above $\lesssim $ is replaced by $\succ $.
\end{definition}

\cite{DDGO20} proved an incompatibility result among SI, No
Future Bias, and RATL. The following are standard assumptions on $\gtrsim $
we will use later.

\textbf{Axiom 0} Completeness and Transitivity.

\textbf{Axiom 1} Outcome Monotonicity: For all $x,y\in X$ and $s\in T$, if $%
x>y$ then $\delta _{(x,s)}\succ \delta _{(y,s)}$.

\textbf{Axiom 2} Impatience: For all $x\in X$ and $s,t\in T$, if $t<s$ then $%
\delta _{(x,t)}\succ \delta _{(x,s)}$.

\textbf{Axiom 3} Continuity: For all $p\in \Delta $, the sets $\left\{ q\in
\Delta :p\gtrsim q\right\} $ and $\left\{ q\in \Delta :q\gtrsim p\right\} $
are weakly closed.

\begin{theorem}
\label{DeJ}(Theorem 2 in \cite{DDGO20}) Suppose that $\gtrsim $
admits a GLBU representation, i.e. there exists a function $u$ defined on $X\times T$
and a constant $\pi (\frac{1}{2})$ such that whenever a lottery $p=\frac{1}{2}\delta
_{(x,t)}+\frac{1}{2}\delta _{(x^{\prime },t^{\prime })}$ with $u(x,t)\geq
u(x^{\prime },t^{\prime })$, $p$ is evaluated by
\begin{equation*}
V(p)=\pi (\frac{1}{2})u(x,t)+\left( 1-\pi (\frac{1}{2})\right) u(x^{\prime
},t^{\prime }).
\end{equation*}%
Then Axioms 0-2, SI, and No Future Bias imply weak RSTL.
\end{theorem}

\bigskip

GLBU is a common feature in a wide scope of behavioral models. However, we
demonstrate that if we stick to No Future Bias, then the incompatibility
between SI and RATL extends to more general models. More specifically, we
relax GLBU to a weaker axiom.

\begin{definition}
A preference relation $\gtrsim $ is said to satisfy \textbf{Weak Certainty
Independence} (WCI) if for any%
\begin{equation*}
p=\delta _{(x,t)},q=\delta _{(y,s)},r=\delta _{(z,u)},
\end{equation*}%
we have 
\begin{equation}
p\gtrsim q\Longrightarrow \frac{1}{2}p+\frac{1}{2}r\gtrsim \frac{1}{2}q+%
\frac{1}{2}r.  \label{WCI}
\end{equation}
\end{definition}

WCI is strictly weaker than the Certainty Independence introduced by \cite{GS89}, which allows $p,q$ in (\ref{WCI}) to be any lottery not
necessarily certain. Moreover, WCI is strictly weaker than the Axiom of
Degenerate Independence studied by \cite{GKP92}, which
does not require $r$ to be a certain prize. Clearly, if a model admits a GLBU representation, then it satisfies WCI.

Our incompatibility result is as follows:

\begin{theorem}
\label{incompatibility}Suppose $\gtrsim $ satisfies Axioms 0-3, SI, and No
Future Bias. Then for any $\left( x,t\right)$ in $\inte(X\times T)$, there
exists $s>0$ such that for any $\tau \in (0,s)$, we have 
\[\frac{1}{2}\delta
_{(x,t-\tau )}+\frac{1}{2}\delta _{(x,t+\tau )}\gtrsim \delta _{(x,t)}.\] In
other words, SI and No Future Bias imply weak RSTL in the local sense.
\end{theorem}

Although the conclusion of Theorem \ref{incompatibility} is local in
contrast to everywhere RSTL in Theorem \ref{DeJ}, our incompatibility
theorem includes substantially more models than DeJarnette's et al. To see
this, we note that WCI means a binary lottery $\frac{1}{2}\delta _{(x,t)}+%
\frac{1}{2}\delta _{(y,s)}$ is evaluated by 
\begin{equation*}
f(u(x,t),u(y,s))
\end{equation*}%
for some utility function $u$ evaluating fixed prizes and some $f$ that
aggregates these two possibilities. Besides monotonicity, there is no
restrictions on $f$ in the models covered by Theorem \ref{incompatibility}. On the other hand, GLBU requires $f$ takes the form
\begin{equation*}
f(m,n)=a\max u(m,n)+(1-a)\min u(m,n)
\end{equation*}%
for some constant $a\in (0,1)$.

The flexibility of $f$ resulting from relaxing GLBU to WCI allows us to
apply the incompatibility result to more models. In reference point models,
individuals' evaluation of an outcome can depend on other factors such as
their evaluation of other possible outcomes. For example, \cite{LS86} proposed a disappointment model in which individuals compare the
actual outcome to a prior expectation. When adapting to our setting, their
model becomes%
\begin{equation*}
V(p)=\mathbb{E}_{p}\left[ u(x,t)+R(u(x,t)-\bar{u})\right] ,
\end{equation*}%
where $\bar{u}$ denotes the prior expectation and $R$ is a increasing
function satisfying $R(0)=0$. This disappointment model rarely satisfies
GLBU. However, with many natural choices of $\bar{u}$ (e.g. $\bar{u}$ is a
constant or $\bar{u}=\mathbb{E}_{p}u$), the disappointment model satisfies
WCI.

\section{Accommodation of RATL with SI}

Given that our incompatibility result that says No Future Bias and SI cannot
be accommodated with RATL under a very weak form of Independence, we are
motivated to drop No Future Bias to accommodate RATL with SI.

To begin with, the following is an example that satisfies SI and RATL in a
strict sense everywhere.

\begin{example}
\label{ex}Suppose $\gtrsim $ can be represented by 
\begin{equation*}
V(p)=\mathbb{E}_{p}-\left[ -\log \left( d^{t^{a}}v(x)\right) \right] ^{b},
\end{equation*}%
where $a>1$, $b\in \left( \frac{1}{a},1\right) $, $d\in (0,1)$ and $v$ is
monotone increasing and $Range(v)\in (0,1)$. Then $\gtrsim $ satisfies Axioms
0-3, and Independence. Moreover, it satisfies strict SI and strict RATL,
in the sense that%
\begin{equation*}
\frac{1}{2}\delta _{(x_{1},t_{1})}+\frac{1}{2}\delta _{(x_{2},t_{2})}\succ 
\frac{1}{2}\delta _{(x_{1},t_{2})}+\frac{1}{2}\delta _{(x_{2},t_{1})}
\end{equation*}%
for any $t_{1}<t_{2},x_{1}>x_{2}$ and 
\begin{equation*}
\delta _{\left( x,\bar{t}_{p_{x}}\right) }\succ p_{x}
\end{equation*}%
for any time lottery $p_{x}$.
\end{example}

The example lies in a large class of utility functions that satisfy SI and
RATL everywhere. We prove a representation theorem that characterizes this
class of functions in Expected Utility.

\textbf{Axiom 4} Independence: For all $p,q,r\in \Delta (X\times T)$ and $%
\lambda \in (0,1)$,%
\begin{equation*}
p\gtrsim q\Leftrightarrow \lambda p+(1-\lambda )r\gtrsim \lambda
q+(1-\lambda )r
\end{equation*}

\textbf{Axiom 5} Double Cancellation: For any $x_{1},x_{2},x_{3}\in X$ and $%
t_{1},t_{2},t_{3}\in T$, if $\delta _{(x_{1},t_{1})}\gtrsim \delta
_{(x_{2},t_{2})}$, and $\delta _{(x_{2},t_{3})}\gtrsim \delta
_{(x_{3},t_{1})}$, then $\delta _{(x_{1},t_{3})}\gtrsim \delta
_{(x_{3},t_{2})}$.

\begin{theorem}
\label{multiplicative}$\gtrsim $ satisfies Axioms 0-5, RATL, and SI if and
only if $\gtrsim $ can be represented by 
\begin{equation*}
V(p)=\mathbb{E}_{p}[\phi \left( D(t)v(x)\right) ]
\end{equation*}%
and the following are satisfied:

(1) $\phi $ is strictly increasing, and continuous on $Range(D\cdot v)$. In
addition, $\phi \circ \exp $ is convex.

(2) $D>0$ is strictly decreasing and continuous on $T$.

(3) $v>0$ is strictly increasing and continuous on $X$.

(4) $\phi \left( D(\cdot )v(x)\right) $ is concave for each $x\in X$.

What's more, $\ln D$ and $\ln v$ are unique up to positive linear
transformations. After fixing a choice of $D$ and $v$, $\phi $ is unique up to positive
linear transformations.
\end{theorem}

\begin{remark}
Given the representation%
\begin{equation*}
V(p)=\mathbb{E}_{p}[\phi \left( D(t)v(x)\right)] ,
\end{equation*}%
RATL arises from two parts: Future Bias and curvature of $\phi $. Due to the
separability of time with monetary prizes, Future Bias is equivalent to
increasing discount rate. Hence, for any $\tau >0$, the discount from $%
\delta _{(x,t)}$ to $\delta _{(x,t+\tau )}$ is less than the discount from $%
\delta _{(x,t+\tau )}$ to $\delta _{(x,t+2\tau )}$. However, Future Bias
alone is usually not enough to explain RATL: 
\begin{equation*}
V(p)=\mathbb{E}[e^{-t-\frac{t^{3}}{3}}v(x)]
\end{equation*}%
satisfies Future Bias, but not RATL. Only if the discount function $D$ is
concave can $V(p)=\mathbb{E}[D(t)v(x)]$ satisfies RATL, but this is usually
not desirable.

Risk attitude is capture by $\phi $. The more concave $\phi $ is, the more
risk averse decision makers are and the more likely they are RATL. However,
SI requires $\phi $ to be more convex than logarithmic functions. Example %
\ref{ex} shows that $\phi (x):=-(-\log x)^{b}$ is at the right level of
convexity that is able to accommodate both SI and RATL.
\end{remark}

Despite discarding No Future Bias and adopting everywhere RATL in the
representation theorem, it is important to note that we are not arguing No
Future Bias is completely false, or RATL should hold everywhere. Admittedly,
there is evidence in support of No Future Bias and lab experiments do document some RSTL choices when there is a small probability of a very early arrival. However, our
representation shows that RATL and SI are not fundamentally conflicting in
Expected Utility. Hence, decision makers' occasional RATL choices may be
due to Future Bias but not the failure of Expected Utility or SI. This opens the door for future study.

\section*{Appendix}

\addcontentsline{toc}{section}{Appendices} \renewcommand{\thesubsection}{%
\Alph{subsection}}

\subsection{Key lemmas about SI and RATL}

Assuming Independence and Continuity, preferences on $\Delta \left( X\times T\right) $ can
be represented by 
\begin{equation*}
V(p)=\mathbb{E}_{p}u(x,t)
\end{equation*}%
for some $u$ defined on $X\times T$. Within the Expected Utility framework,
RATL and RSTL relations can be passed to the curvature of $u$.

\begin{lemma}
\label{concaveRATL}In the Expected Utility model in which a preference $%
\gtrsim $ is represented by 
\begin{equation*}
V(p)=\mathbb{E}_{p}u(x,t),
\end{equation*}%
$\gtrsim $ is RATL if and only if $u\left( x,\cdot \right) $ is concave for
any $x\in X$. Similarly, $\gtrsim $ is RSTL if and only if $u\left( x,\cdot
\right) $ is convex for any $x$.
\end{lemma}

This claim follows directly from Jensen's inequality.

There is a neat characterization of SI within Expected Utility: time and
monetary prizes are complimentary. This has been observed by \cite{DDGO20}, but here we give a rigorous proof. 

\begin{proposition}
\label{derivativeSI}Suppose $\gtrsim $ can be represented by 
\begin{equation*}
V(p)=\mathbb{E}_{p}u(x,t),
\end{equation*}%
where $u$ is twice differentiable. Then the relation $\gtrsim $ satisfies SI
if and only if $\frac{\partial u}{\partial x\partial t}\leq 0$ on $%
\inte(X\times T)$.
\end{proposition}

\begin{proof}
$\qquad \Longrightarrow $ Stochastic Impatience implies that for any $%
x_{0}\in int\left( X\right) $, $t_{0}\in int\left( T\right) $, and small
increments $\Delta x>0$ and $\Delta t>0$, we have 
\begin{equation}
u(x_{0}+\Delta x,t_{0}-\Delta t)+u(x_{0}-\Delta x,t_{0}+\Delta t)\geq
u(x_{0}+\Delta x,t_{0}+\Delta t)+u(x_{0}-\Delta x,t_{0}-\Delta t).
\label{eqn1}
\end{equation}%
Consider the second-order Taylor series of $f$ at $\left( x_{0},t_{0}\right) 
$ and let $\Delta x=\Delta t\rightarrow 0$. By slight abuse of notation, we
write $u_{x}\left( x_{0},t_{0}\right) $ as $u_{x}$ and employ similar
abbreviations for other derivatives at $(x_{0},t_{0})$. Rewrite Equation (\ref%
{eqn1}) as%
\begin{eqnarray*}
&&u(x_{0},t_{0})+\Delta x\cdot u_{x}-\Delta t\cdot u_{t}+\frac{1}{2}\left(
u_{xx}\left( \Delta x\right) ^{2}+u_{tt}\left( \Delta t\right)
^{2}-2u_{xt}\left( \Delta x\right) \left( \Delta t\right) \right) \\
&&+u(x_{0},t_{0})-\Delta x\cdot u_{x}+\Delta t\cdot u_{t}+\frac{1}{2}\left(
u_{xx}\left( \Delta x\right) ^{2}+u_{tt}\left( \Delta t\right)
^{2}-2u_{xt}\left( \Delta x\right) \left( \Delta t\right) \right) +o((\Delta
x)^{3}) \\
&\geq &u(x_{0},t_{0})+\Delta x\cdot u_{x}+\Delta t\cdot u_{t}+\frac{1}{2}%
\left( u_{xx}\left( \Delta x\right) ^{2}+u_{tt}\left( \Delta t\right)
^{2}+2u_{xt}\left( \Delta x\right) \left( \Delta t\right) \right) \\
&&+u(x_{0},t_{0})-\Delta x\cdot u_{x}-\Delta t\cdot u_{t}+\frac{1}{2}\left(
u_{xx}\left( \Delta x\right) ^{2}+u_{tt}\left( \Delta t\right)
^{2}+2u_{xt}\left( \Delta x\right) \left( \Delta t\right) \right) +o((\Delta
x)^{3})
\end{eqnarray*}%
Equivalently, 
\begin{equation*}
u_{xt}\leq 0.
\end{equation*}%
$\qquad \Longleftarrow $ For any $t_{1}<t_{2}$ in $T$ and $x_{1}>x_{2}$ in $%
X $, we have 
\begin{eqnarray*}
&&u(x_{1},t_{1})+u(x_{2},t_{2})-u(x_{1},t_{2})-u(x_{2},t_{1}) \\
&=&\left( u(x_{1},t_{1})-u(x_{1},t_{2})\right) +\left(
u(x_{2},t_{2})-u(x_{2},t_{1})\right) \\
&=&-\int_{t_{1}}^{t_{2}}u_{t}\left( x_{1},t\right)
dt+\int_{t_{1}}^{t_{2}}u_{t}\left( x_{2},t\right) dt \\
&=&\int_{t_{1}}^{t_{2}}\left( u_{t}\left( x_{2},t\right) -u_{t}\left(
x_{1},t\right) \right) dt.
\end{eqnarray*}%
Since $u_{xt}\leq 0$, the integrand of the last line $u_{t}\left(
x_{2},t\right) -u_{t}\left( x_{1},t\right) \geq 0$ so the integral is
non-negative.
\end{proof}

\subsection{Analysis of Example \protect\ref{ex}}

We shall prove that $\gtrsim $ in Example \ref{ex} satisfies strict SI and
strict RATL.

To check that $\gtrsim $ satisfies strict SI, we adopt Proposition \ref%
{derivativeSI} (use strict SI and strict inequality in the second order
derivative in place of SI and weak inequality). We have%
\begin{eqnarray*}
&&\frac{\partial }{\partial t\partial x}\left( -\left[ -\log \left(
d^{t^{a}}v(x)\right) \right] ^{b}\right) \\
&=&\frac{\partial }{\partial t\partial x}\left( -\left[ -t^{a}\log d-\log
v(x)\right] ^{b}\right) \\
&=&\frac{\partial }{\partial x}\left( -b\left[ -t^{a}\log d-\log v(x)\right]
^{b-1}\left( -at^{a-1}\log d\right) \right) \\
&=&-b(b-1)\left[ -t^{a}\log d-\log v(x)\right] ^{b-2}\left( -at^{a-1}\log
d\right) (-\frac{1}{v(x)}v^{\prime }(x)) \\
&<&0
\end{eqnarray*}%
because $b(b-1)<0,-t^{a}\log d-\log v(x)>0,-at^{a-1}\log d>0$ and $-\frac{1}{%
v(x)}v^{\prime }(x)<0$.

As for RATL, we have%
\begin{eqnarray*}
&&\frac{\partial }{\partial t\partial t}\left( -\left[ -\log \left(
d^{t^{a}}v(x)\right) \right] ^{b}\right) \\
&=&\frac{\partial }{\partial t}\left( -b\left[ -t^{a}\log d-\log v(x)\right]
^{b-1}\left( -at^{a-1}\log d\right) \right) \\
&=&-b(b-1)\left[ -t^{a}\log d-\log v(x)\right] ^{b-2}\left( -at^{a-1}\log
d\right) ^{2} \\
&&-b\left[ -t^{a}\log d-\log v(x)\right] ^{b-1}\left( -a(a-1)t^{a-2}\log
d\right) \\
&=&-b\left[ -t^{a}\log d-\log v(x)\right] ^{b-2}\left( -at^{a-1}\log
d\right) \left[ -(b-1)at^{a-1}\log d+\left[ -t^{a}\log d-\log v(x)\right]
(a-1)t^{-1}\right] .
\end{eqnarray*}%
Since 
\begin{equation*}
-b\left[ -t^{a}\log d-\log v(x)\right] ^{b-2}\left( -at^{a-1}\log d\right) <0
\end{equation*}%
and 
\begin{eqnarray*}
&&-(b-1)at^{a-1}\log d+\left[ -t^{a}\log d-\log v(x)\right] (a-1)t^{-1} \\
&>&-(b-1)at^{a-1}\log d+(-t^{a}\log d)(a-1)t^{-1} \\
&=&t^{a-1}((-(b-1)a-(a-1))\log d \\
&=&t^{a-1}(-ab+1)\log d \\
&>&0,
\end{eqnarray*}%
we have%
\begin{equation*}
\frac{\partial }{\partial t\partial t}\left( -\left[ -\log \left(
d^{t^{a}}v(x)\right) \right] ^{b}\right) <0.
\end{equation*}%
By Lemma \ref{concaveRATL} (with strict RATL and strictly concave in place
of RATL and concave), $\gtrsim $ satisfies strict RATL.

\subsection{Proof of Theorem \protect\ref{incompatibility}}

\begin{lemma}
Suppose $\gtrsim $ satisfies Continuity, Impatience and Outcome
Monotonicity. Fix $\left( x,t\right) \in \inte(X\times T)$. Then there exists $s>0$
such that for any $\tau \in (0,s)$, there exists $y<x$ in $X$ satisfying $%
\delta _{(y,t-\tau )}\sim \delta _{(x,t)}$.
\end{lemma}

\begin{proof}
Let $w=\min X$ be the worst possible monetary prize. Let $S:=\{c:\delta
_{(w,c)}\gtrsim \delta _{(x,t)}\}$ and $L:=\{c:\delta _{(w,c)}\lesssim
\delta _{(x,t)}\}$. Then Continuity and Impatience imply that $S$ and $L$
are closed intervals and $S\cup L=X$. By monotonicity, $\delta _{(w,t)}\prec
\delta _{(x,t)}$ so $t\in L$.

\textbf{Case 1}: $S\neq \varnothing $. Let $d:=\max S=\min L$. Then $\delta
_{(w,d)}\sim \delta _{(x,t)}$. We must have $d<t$; otherwise, Impatience and
Outcome Monotonicity would imply%
\begin{equation*}
\delta _{(x,t)}\gtrsim \delta _{(x,d)}\succ \delta _{(w,d)},
\end{equation*}%
a contradiction.

Set $s=t-d$. For any $\tau \in (0,s)$, we let $S_{\tau }:=\left\{ y:\delta
_{(y,t-\tau )}\gtrsim \delta _{(x,t)}\right\} $ and $L_{\tau }:=\{y:\delta
_{(y,t-\tau )}\lesssim \delta _{(x,t)}\}$. Then $S_{\tau }$ and $L_{\tau }$
are closed intervals with $S_{\tau }\cup L_{\tau }=X$. By Impatience, $x\in
S_{\tau }$. What's more, $\delta _{(w,t-\tau )}\lesssim \delta _{(w,d)}\sim
\delta _{(x,t)}$ so $w\in L_{\tau }$. We set $y_{\tau }:=\min S_{\tau }=\max
L_{\tau }$. Then it follows that $\delta _{(y_{\tau },t-\tau )}\sim \delta
_{(x,t)}$. Similar to before, Impatience and Outcome Monotonicity imply $%
y_{\tau }<x$.

\textbf{Case 2}: $S=\varnothing $. Set $s=t-\min T$. For any $\tau \in (0,s)$, we
let $S_{\tau }:=\left\{ y:\delta _{(y,t-\tau )}\gtrsim \delta
_{(x,t)}\right\} $ and $L_{\tau }:=\{y:\delta _{(y,t-\tau )}\lesssim \delta
_{(x,t)}\}$. Then $x\in S_{\tau }$. In addition, $S=\varnothing $ implies $%
\delta _{(x,t)}\succ \delta _{(w,t-\tau )}$. Thus, $w\in L_{\tau }$. We set $%
y_{\tau }=\min S_{\tau }=\max L_{\tau }$. Then $\delta _{(y_{\tau },t-\tau
)}\sim \delta _{(x,t)}$ and $y_{\tau }<x$.
\end{proof}

Fix $(x,t)\in $int$(X\times T)$, and let $s$ be as in the previous lemma.
For any $\tau \in (0,s)$, choose $y<x$ in $X$ satisfying $\delta _{(y,t-\tau
)}\sim \delta _{(x,t)}$. No Future Bias implies%
\begin{equation*}
\delta _{(x,t+\tau )}\gtrsim \delta _{(y,t)}.
\end{equation*}%
By WCI, we are able to add the same term on both sides:%
\begin{equation}
\frac{1}{2}\delta _{(x,t+\tau )}+\frac{1}{2}\delta _{(x,t-\tau )}\gtrsim 
\frac{1}{2}\delta _{(y,t)}+\frac{1}{2}\delta _{(x,t-\tau )}\text{.}
\label{incompatibility 1}
\end{equation}

By SI, 
\begin{equation}
\frac{1}{2}\delta _{(y,t)}+\frac{1}{2}\delta _{(x,t-\tau )}\gtrsim \frac{1}{2%
}\delta _{(x,t)}+\frac{1}{2}\delta _{(y,t-\tau )}.  \label{incompatibility 2}
\end{equation}%
The two terms on the right hand side can be combined into 
\begin{equation}
\frac{1}{2}\delta _{(x,t)}+\frac{1}{2}\delta _{(y,t-\tau )}\sim \delta
_{(x,t)}  \label{incompatibility 3}
\end{equation}%
using WCI.

Equation (\ref{incompatibility 1}), (\ref{incompatibility 2}), and (\ref%
{incompatibility 3}) then imply%
\begin{equation*}
\frac{1}{2}\delta _{(x,t+\tau )}+\frac{1}{2}\delta _{(x,t-\tau )}\gtrsim
\delta _{(x,t)}
\end{equation*}%
as desired.

\subsection{Proof of Theorem \protect\ref{multiplicative}}

We begin by writing multiplicative forms into additive forms.

\begin{lemma}
\label{additive}$\gtrsim $ satisfies Axioms 0-5, RATL, and SI if and only if $%
\gtrsim $ can be represented by 
\begin{equation*}
V(p)=\mathbb{E}_{p}[\phi \left( D(t)+v(x)\right) ]
\end{equation*}%
and the following are satisfied:

(1') $\phi $ is strictly increasing, convex and continuous on $Range(D+v)$.

(2') $D$ is strictly decreasing and continuous on $T$.

(3') $v$ is strictly increasing and continuous on $X$.

(4') $\phi \left( D(\cdot )+v(x)\right) $ is concave for each $x\in X$.

What's more, such $D$ and $v$ are unique up to positive linear
transformations. After fixing a choice of $D$ and $v$, $\phi $ is unique up
to positive linear transformations.
\end{lemma}

\begin{proof}
$\qquad \Longleftarrow $ Suppose we have such a representation of $\gtrsim $.
Completeness, Transitivity, and Independence are clearly satisfied. Since $%
\phi $ is strictly increasing and $D$ is strictly decreasing in time, we
have Impatience. Outcome Monotonicity follows from strictly increasing $\phi 
$ and strictly increasing $v$. Continuity of $\gtrsim $ follows from the
continuity of $D$ and $v$. By Lemma \ref{concaveRATL}, concavity of $%
\phi \left( D(\cdot )+v(x)\right) $ implies RATL.

To check Double Cancellation, let $x_{1},x_{2},x_{3}\in X$ and $%
t_{1},t_{2},t_{3}\in T$ such that $\delta _{(x_{1},t_{1})}\gtrsim \delta
_{(x_{2},t_{2})}$, and $\delta _{(x_{2},t_{3})}\gtrsim \delta
_{(x_{3},t_{1})}$. We notice that because of monotonicity of $\phi $, 
\begin{equation*}
\phi (D(t_{1})+v(x_{1}))\geq \phi (D(t_{2})+v(x_{2}))
\end{equation*}%
implies 
\begin{equation*}
D(t_{1})+v(x_{1})\geq D(t_{2})+v(x_{2}).
\end{equation*}%
Similarly, 
\begin{equation*}
D(t_{3})+v(x_{2})\geq D(t_{1})+v(x_{3})
\end{equation*}%
follows from $\delta _{(x_{2},t_{3})}\gtrsim \delta _{(x_{3},t_{1})}.$
Combining the last two equations, we get%
\begin{equation*}
D(t_{3})+v(x_{1})\geq D(t_{2})+v(x_{3}),
\end{equation*}%
which is equivalent to $\delta _{(x_{1},t_{3})}\gtrsim \delta
_{(x_{3},t_{2})}.$

The last property we want to check is SI. If $x_{1}>x_{2}$ and $t_{1}<t_{2}$%
, then 
\begin{equation*}
\lbrack
D(t_{1})+v(x_{1})]+[D(t_{2})+v(x_{2})]=[D(t_{1})+v(x_{2})]+[D(t_{2})+v(x_{1})].
\end{equation*}%
What's more, $[D(t_{1})+v(x_{1})]\geq \lbrack
D(t_{1})+v(x_{2})],[D(t_{2})+v(x_{1})]\geq \lbrack D(t_{2})+v(x_{2})]$.
Convexity of $\phi $ implies 
\begin{equation*}
\phi \left( D(t_{1})+v(x_{1})\right) +\phi \left( D(t_{2})+v(x_{2})\right)
\geq \phi \left( D(t_{1})+v(x_{2})\right) +\phi \left(
D(t_{2})+v(x_{1})\right) .
\end{equation*}%
$\qquad \Longrightarrow $ Suppose $\gtrsim $ satisfies Axioms 0-5, RATL, and
SI. Let $X^{\ast }$ denote $\left\{ \delta _{(x,t)}:(x,t)\in X\times T\right\} $, viewed as
a subset of $\Delta (X\times T)$. $X^{\ast }$ has the preference relation $%
\gtrsim _{X^{\ast}}$ inherited from $\Delta (X\times T)$. Then $\gtrsim _{X^{\ast}}$
satisfies Outcome Monotonicity, Impatience, Continuity and Double
Cancellation. By \cite{De59} or \cite{Fi70}, there exist a continuous
function $D$ on $T$ and a continuous function $v$ on $X$ such that $\gtrsim _{X^{\ast}}$ can
be represented by $D(t)+v(x)$. What's more, $D$ and $v$ are unique up to
positive linear transformations. Outcome Monotonicity and Impatience of $%
\gtrsim _{X^{\ast}}$ implies that $v$ is strictly increasing and $D$ is strictly
decreasing. 

We fix a choice of $D$ and $v$. Let $Y:=Range(D+v)$. Since $D$ and $v$ are
continuous functions on closed intervals, $Y$ is a connected interval. $\gtrsim $
induces a preference relation $\gtrsim _{Y}$ on $\Delta Y$, the set of simple
lotteries on $Y$, as follows. For $p\in \Delta (X\times T)$, define $%
p^{\prime }$ by 
\begin{equation*}
p^{\prime }(a)=p(\left\{ (x,t):D(t)+v(x)=a\right\} )
\end{equation*}%
for $a\in Y$. Let $p^{\prime }\gtrsim _{Y}q^{\prime }$ if and only if $%
p\gtrsim q$. Independence gurantees the well-definedness of $\gtrsim _{Y}$.
It can be easily checked that $\gtrsim _{Y}$ is a preference relation and
inherits Independence and Continuity from $\gtrsim $. Thus, there exists a
continuous function $\phi :Y\rightarrow \mathbb{R}$ such that $\gtrsim _{Y}$
can be represented by 
\begin{equation*}
V(p^{\prime })=\mathbb{E}_{p^{\prime }}\phi \text{.}
\end{equation*}%
What's more, $\phi $ is unique up to positive linear transformations. It
follows that $\gtrsim $ can be represented by 
\begin{equation*}
V(p)=\mathbb{E}_{p}[\phi \left( D(t)+v(x)\right) ].
\end{equation*}%
Since $v$ is strictly increasing, Outcome Monotonicity implies that $\phi $
is strictly increasing.

By Lemma \ref{concaveRATL}, RATL implies concavity of $\phi
\left( D(\cdot )+v(x)\right) $ for each $x\in X$.

We then demonstrate convexity of $\phi $ follows from SI. Let $A=\sup v-\inf v>0$,
possibly infinite. We first show that for any $0<\varepsilon <A$, $\phi
(b+\varepsilon )-\phi (b)$ is strictly increasing in $b$ whenever $%
b,b+\varepsilon \in Range(D+v)$. This is because by SI, for any $x_{1},x_{2}$
such that $v(x_{1})=v(x_{2})+\varepsilon $, we have 
\begin{equation*}
\phi (D(t_{1})+v(x_{1}))-\phi (D(t_{1})+v(x_{2}))\geq \phi
(D(t_{2})+v(x_{1}))-\phi (D(t_{2})+v(x_{2}))
\end{equation*}%
for any $t_{1}<t_{2}$. That is $\phi (\cdot +v(x_{2})+\varepsilon )-\phi
(\cdot +v(x_{2}))$ is increasing in $Range(D)$. As $x_{1},x_{2}$ ranges in $X
$, $v(x_{2})$ covers $Range(v)\cap (Range(v)-\varepsilon )$ so $\phi (\cdot
+\varepsilon )-\phi (\cdot )$ is increasing in $Range(D+v)\cap
(Range(D+v)-\varepsilon )$.

We claim that for any rational $a\in (0,1)$ and any $c_{2}>c_{1}$, $%
c_{1},c_{2}\in Range(D+v)$, we have 
\begin{equation*}
a\phi (c_{1})+(1-a)\phi (c_{2})\geq \phi (ct_{1}+(1-a)c_{2}).
\end{equation*}%
Let $a=\frac{p}{q}$, where $p,q$ are positive integers such that $\delta :=\frac{%
c_{2}-c_{1}}{q}<A$. Then $ac_{1}+(1-a)c_{2}=c_{1}+(q-p)\delta $, and 
\begin{equation*}
\phi (k\delta +c_{1})-\phi ((k-1)\delta +c_{1})
\end{equation*}%
is increasing in $k$. Thus, 
\begin{eqnarray*}
&&\frac{1}{q-p}\left[ \phi (ac_{1}+(1-a)c_{2})-\phi (c_{1})\right]  \\
&=&\frac{1}{q-p}\sum_{k=1}^{q-p}\left[ \phi (k\delta +c_{1})-\phi
((k-1)\delta +c_{1})\right]  \\
&\leq &\phi (ac_{1}+(1-a)c_{2})-\phi (ac_{1}+(1-a)c_{2}-\delta ) \\
&\leq &\frac{1}{p}\sum_{k=q-p+1}^{q}\left[ \phi (k\delta +c_{1})-\phi
((k-1)\delta +c_{1})\right]  \\
&=&\frac{1}{p}[\phi (c_{2})-\phi (ac_{1}+(1-a)c_{2}]
\end{eqnarray*}%
Reorganizing this, we get
\begin{equation*}
a\phi (c_{1})+(1-a)\phi (c_{2})\geq \phi (ac_{1}+(1-a)c_{2}).
\end{equation*}

We then remove the assumption that $a$ is rational in the last equation by
continuity of $\phi $. Suppose $a$ can be approximated by a sequence of
rationals $\left\{ a_{i}\right\} \subset (0,1)$. For each $i$, we have 
\begin{equation*}
a_{i}\phi (c_{1})+(1-a_{i})\phi (c_{2})\geq \phi (a_{i}c_{1}+(1-a_{i})c_{2}).
\end{equation*}%
Let $i\rightarrow \infty $. Since $\phi $ is continuous, we still have%
\begin{equation*}
a\phi (c_{1})+(1-a)D(c_{2})\geq \phi (ac_{1}+(1-a)c_{2})
\end{equation*}%
in the limit. In conclusion, $\phi $ is convex on $Range(D+v)$.
\end{proof}

\bigskip
Now we go back to the proof of Theorem \ref{multiplicative}.

$\Longrightarrow $ Suppose $\gtrsim $ satisfies Axioms 0-5, RATL, and SI. By
Lemma \ref{additive}, there exist $\phi ^{\ast }$, $D^{\ast }$ and $v^{\ast }
$ such that $\gtrsim $ can be represented by 
\begin{equation*}
V(p)=\mathbb{E}_{p}[\phi ^{\ast }\left( D^{\ast }(t)+v^{\ast }(x)\right) ]
\end{equation*}%
such that properties (1')-(4') in the statement of that lemma hold. We let 
\begin{eqnarray*}
\phi  &=&\phi ^{\ast }\circ \ln  \\
D &=&\exp (D^{\ast }) \\
v &=&\exp (v^{\ast })
\end{eqnarray*}%
then 
\begin{equation*}
V(p)=\mathbb{E}_{p}[\phi \left( D(t)v(x)\right) ],
\end{equation*}%
(1)-(4) in Theorem \ref{multiplicative} hold. The uniqueness of $D,v,$ and $%
\phi $ also follows from the uniqueness of $D^{\ast },v^{\ast },$ and $\phi
^{\ast }$.

$\Longleftarrow $ Conversely, we assume $\gtrsim $ can be represented by 
\begin{equation*}
V(p)=\mathbb{E}_{p}[\phi \left( D(t)v(x)\right) ]
\end{equation*}%
such that (1)-(4) in Theorem \ref{multiplicative} are true. We set 
\begin{eqnarray*}
\phi ^{\ast } &=&\phi \circ \exp  \\
D^{\ast } &=&\ln D \\
v^{\ast } &=&\ln v.
\end{eqnarray*}%
Then (1')-(4') in Lemma \ref{additive} follow from (1)-(4). Hence by Lemma %
\ref{additive}, Axioms 0-5, RATL, and SI are satisfied.

\bibliographystyle{aer}
\bibliography{main.bbl}

\end{document}